\documentclass[onefignum,onetabnum]{siamart171218}

\ifpdf
  \DeclareGraphicsExtensions{.eps,.pdf,.png,.jpg}
\else
  \DeclareGraphicsExtensions{.eps}
\fi

\usepackage{amssymb}
\usepackage{amsmath}
\usepackage{mathtools}
\usepackage{bbm}
\usepackage[utf8]{inputenc}
\usepackage{graphicx}
\usepackage{nameref}
\usepackage{tikz}
\usepackage{wrapfig}
\usepackage{epigraph}
\usepackage{booktabs}
\usepackage{multirow}
\usepackage{longtable}

\usepackage{tikz}
\usetikzlibrary{arrows.meta,patterns}
\usetikzlibrary{ipe} %

\newcommand{\Ex}[2][ ]{\ensuremath{\mathbb{E}_{#1}\left[#2\right]}}

\renewcommand{\Pr}[2][ ]{\ensuremath{\mathbb{P}_{#1}\left(#2\right)}}

\newcommand{\eqdef}{\overset{\Delta}{=}}

\newcommand{\Norm}[1]{\ensuremath{\mathcal{N}\left(#1\right)}}
\newcommand{\CNorm}[1]{\mathcal{CN}\left(#1\right)}

\newcommand{\R}{\mathbb{R}}
\newcommand{\C}{\mathbb{C}}
\newcommand{\Z}{\mathbb{Z}}
\DeclarePairedDelimiterX{\norm}[1]{\lVert}{\rVert}{#1}
\DeclarePairedDelimiterX{\normi}[1]{\lVert}{\rVert_{\infty}}{#1}

\newcommand{\simiid}{\overset{\textrm{i.i.d.}}{\sim}}

\DeclareMathOperator{\ReOp}{Re}
\DeclareMathOperator{\erf}{erf}
\renewcommand{\Re}[1]{\ReOp\left[#1\right]}

\DeclareMathOperator{\ImOp}{Im}
\renewcommand{\Im}[1]{\ImOp\left[#1\right]}

\newcommand{\ejo}{e^{j\Omega}}
\newcommand{\ej}[1]{e^{j {#1}}}

\newtheorem{example}{Example}

\newcommand{\Hk}[1]{\mathcal{H}\!\left(#1\right)}

\newsiamremark{remark}{Remark}
\newsiamremark{hypothesis}{Hypothesis}
\crefname{hypothesis}{Hypothesis}{Hypotheses}
\newsiamthm{claim}{Claim}

\headers{$H_\infty$ Gaussian Processes}{A. Devonport, P. Seiler, and M. Arcak}

\title{Frequency-Domain Gaussian Process Models for $H_\infty$ Uncertainties%
    \thanks{Manuscript for revision.%
\funding{Air Force Office of Scientific
Research grant FA9550-21-1-0288.}
}}

\author{Alex Devonport 
    \thanks{University of California, Berkeley, USA\\ (\email{\{alex\_devonport@berkeley.edu}, \email{arcak@berkeley.edu})}
    \and Peter Seiler\thanks{University of Michigan, Ann Arbor, USA (\email{pseiler@umich.edu})}
    \and Murat Arcak\footnotemark[2]
}

\usepackage{amsopn}

\makeatletter
\newcommand*{\addFileDependency}[1]{%
  \typeout{(#1)}%
  \@addtofilelist{#1}%
  \IfFileExists{#1}{}{\typeout{No file #1.}}%
}
\makeatother

\ifpdf
\hypersetup{
  pdftitle={Frequency-Domain Gaussian Process Models for H-inf Uncertainties},
  pdfauthor={A. Devonport, P. Seiler, and M. Arcak}
}
\fi

\begin{document}

\maketitle

\begin{abstract}
Complex-valued Gaussian processes are commonly used in Bayesian frequency-domain
system identification as prior models for regression. If each realization
of such a process were an $H_\infty$ function with probability one, then the
same model could be used for probabilistic robust control, allowing for
robustly safe learning.
We investigate sufficient conditions for a general complex-domain Gaussian
process to have this property. 
For the special case of processes whose
Hermitian covariance is stationary, we provide an explicit parameterization
of the covariance structure in terms of a summable sequence of nonnegative
numbers.
We then establish how an $H_\infty$ Gaussian process can serve as a prior
for Bayesian system identification and as a probabilistic uncertainty model for
probabilistic robust control. In particular, we compute formulas for refining
the uncertainty model by conditioning on frequency-domain data and for
upper-bounding the probability that the realizations of the process satisfy a
given integral quadratic constraint.
\end{abstract}



Probabilistic models of input-output dynamical systems, where the input-output
relationship itself contains probabilistic elements separate from input noise or
measurement error, have important applications
both in system identification and probabilistic robust control.
In system identification, probabilistic systems act as models of prior belief in
Bayesian estimates of the system dynamics. In robust control, probabilistic
models form the core of probabilistic robustness analysis: the objective is then
to verify, or design a controller such that, the ensemble of uncertainties for
which the system is stabilized has high probability.
These two lines of research-- Bayesian system identification and probabilistic
robust control-- have generally been developed separately, and use different
types of probabilistic models.

A number of models for Bayesian system identification have been developed
in the time domain, with works
like~\cite{pillonetto2010new}
and~\cite{chen2012estimation}
using Gaussian processes to identify the impulse
response of a linear time-invariant (LTI) stable system. 
Subsequent works consider frequency-domain
regression,
such as~\cite{lataire2016transfer}
which uses a modified complex Gaussian process regression
model to estimate transfer functions from discrete Fourier transform (DFT) data,
and~\cite{stoddard2019gaussian}
which considers a
similar regression approach to estimate the generalized frequency response of
nonlinear systems.
The models developed in these works are generally nonparametric (so as not to
place \emph{a priori} restrictions on the system order). However, with the
exception of the time-domain stable kernels of~\cite{pillonetto2022regularized},
their stability properties are generally developed only for the predictive mean.
 
Models for probabilistic robust control have been developed for probabilistic
extension of $\mu$
analysis~(\cite{khatri1998guaranteed,balas2012analysis,biannic2021advanced}),
disk margins~(\cite{somers2022probabilistic}), scenario
optimization~\cite{calafiore2006scenario}, and the methods reviewed
in~\cite{calafiore2007probabilistic}. Unlike for Bayesian system identification,
these models usually place severe structural restrictions on the form of the
unknown model, including a finite (known) bound on the order. These restrictions
ensure that each realization of the uncertain system is guaranteed to represent
a physically interpretable system (e.g. an element of $H_\infty$), allowing for
a meaningful interpretation of probabilistic guarantees of robustness.

Since both Bayesian system identification and probabilistic robust control use
probabilistic uncertainty models, applying both techniques to the same model is
a promising strategy for learning an uncertain control system and asserting
probabilistic behavioral certificates.
The central concept is to use the Bayesian uncertainty of the learned model
to construct a probabilistic robustness guarantee for a suitably
chosen controller. However, this approach requires a probabilistic model that
enjoys the benefits of both classes of models described above, namely that the
model be nonparametric and that realizations are almost surely interpretable. 

This paper develops a class of probabilistic models, the $H_\infty$ Gaussian
processes, which are both nonparametric and almost surely interpretable. In
brief, an $H_\infty$ Gaussian process is a random complex function of a complex
variable, whose gains are complex-normally distributed, and whose sample paths
are $H_\infty$ functions-- that is, interpretable as causal and stable systems-- of
arbitrary (or even infinite) order.
An uncertain system modeled by an $H_\infty$ Gaussian process admits
refinement through data through Gaussian process regression: By conditioning on point
observations in the frequency domain, the model becomes more accurate, though
still uncertain in unobserved frequency ranges.
On the other hand, an $H_\infty$ Gaussian process model admits robustness
analysis by virtue of the fact that it represents an ensemble of
systems, precisely as in robust control, with the additional structure of a
weight (the probability measure) over members of the ensemble. 

This paper provides three main contributions.
The first contribution is to provide a
mathematical foundation for the class of $H_\infty$ Gaussian processes (GP).
Section~\ref{sec:preliminaries} introduces the system setup, reviews
background information on complex-valued random variables and stochastic
processes, and introduces the classes of $H_\infty$ GPs.
Section~\ref{sec:hinf-processes} then investigates how to construct $H_\infty$
Gaussian process. The main result here is a set of sufficient conditions under
which a random complex function of a complex variable will have the $H_\infty$
property.
In addition to the general conditions, we provide a complete characterization
(Theorem~\ref{prop:stationary-process}) of
the covariance structure of a special class of $H_\infty$ Gaussian process,
namely those whose Hermitian covariance is stationary. Each Hermitian stationary
$H_\infty$ process is parameterized by a summable sequence of
nonnegative reals, which lead to computationally tractable closed forms for
certain choices of sequences. 

The second contribution is to establish how to refine an $H_\infty$ GP using data.
Section~\ref{sec:regression}
reviews widely linear and strictly linear complex estimators for complex
Gaussian process regression and presents numerical examples of Bayesian system
identification. 
Contrary to other recent work in Bayesian system
identification, we choose to use the \emph{strictly linear estimator} for our
Gaussian process models instead of the \emph{widely linear estimator}.
Although the widely linear estimate is superior for general processes, we
find that for $H_\infty$ Gaussian process models the strictly linear estimator
works nearly as well while being simpler and more stable to compute than the
widely linear estimator.

The final contribution is to establish probabilistic guarantees of robustness for system
models that use an $H_\infty$ GP as a feedback uncertainty. In
Section~\ref{sec:hinf-gp-robustness}, 
we shall see that establishing accuracy-only probabilistic
guarantees for a number of several robustness certificates-- particularly those based on small-gain arguments and
IQCs--
comes down to establishing an
inequality of the form $\Pr{\normi{f} < u} \ge 1-\delta$ for some $H_\infty$
Gaussian process $f$ with known mean and covariance functions. Bounds of this
form, in turn, can be established by computing the expected number of gain
upcrossings of $f$, which can be carried out by means of Belyaev
formulas~\cite{belyaev1968number}.

In prior work~\cite{devonport2023frequency}, we developed  
initial results for the mathematical foundations of, and Bayesian regression
with, $H_\infty$ GPs. The first two contributions of the present paper are
extensions of this work, while the third is entirely novel.

\subsection{Notation}

For a complex element $X$, $X^*$ denotes the complex conjugate, or Hermitian
transpose where appropriate.
We denote the exterior of the unit disk as
$E = \{Z\in\C : |z| > 1\}$, and its closure as
$\bar{E} = \{Z\in\C : |z| \ge 1\}$.
$L_2$ is the Hilbert space of functions $f:\C\to\C$ such that
$\int_{-\pi}^{\pi} |f( e^{j\Omega})|^2d\Omega < \infty$,
equipped with the
inner product $\langle f,g\rangle_2 = \int_{-\pi}^{\pi}
f(e^{j\Omega}){g^*}(e^{j\Omega})d\Omega$.
$H_2$ is the Hilbert space of functions $f:\C\to\C$ that are bounded
and analytic for all $z\in E$ and
$\int_{-\pi}^{\pi} |f(Re^{j\Omega})|^2d\Omega < \infty$ for $R\ge 1$,
equipped with the
inner product $\langle f,g\rangle_2 = \int_{-\pi}^{\pi}
f(e^{j\Omega}){g^*}(e^{j\Omega})d\Omega$. It is a vector subspace of
$L_2$.
$H_\infty$ is the Banach space of functions $f:\bar{E}\to\C$ that are bounded
and analytic for all $z\in E$ and
$\sup_{\Omega\in[-\pi,\pi]} |f(e^{j\Omega})| < \infty$, equipped with the
norm $\normi{f} = \sup_{\Omega\in[-\pi,\pi]} |f(e^{j\Omega})|$.
$\ell^1$ is the space of absolutely summable sequences,
that is sequences $\{a_n\}_{n=0}^\infty$
such that $\sum_{n=0}^\infty |a_n| <\infty$.
$x\sim \Norm{\mu,\Sigma}$ denotes a Gaussian distribution with mean $\mu=\Ex{x}$ and
covariance $\Sigma=\Ex{(x-\mu)^2}$; likewise, $w\sim \CNorm{\mu,\Sigma,\tilde{\Sigma}}$ denotes a
complex Gaussian distribution with mean $\mu=\Ex{w}$, Hermitian covariance $\Sigma = \Ex{|w-\mu|^2}$,
and complementary covariance $\tilde{\Sigma}=\Ex{(w-\mu)^2}$.

\section{Preliminaries}
\label{sec:preliminaries}

$H_\infty$ Gaussian processes are nonparametric statistical models for
causal, LTI, BIBO stable systems in the frequency domain. Since our main focus
will be the probabilistic aspects of the model, we restrict our attention to the
simplest dynamical case: a single-input single-output system in discrete time.
Thus, our dynamical systems are frequency-domain multiplier operators
$H_f:L_2\to L_2$ whose output is defined pointwise as $(H_f
u)(\omega)=f(\omega)u(\omega)$,
where $f:\C\to\C$ is the system's transfer function. Thanks to the bijection
$H_f \leftrightarrow f$, we generally mean the function $f$ when we refer
to ``the system''.

Since our aim is to construct a probabilistic model for the system that is not
restricted to a finite number of parameters, we must work directly with random
complex functions of a complex variable: this is a special type of complex
stochastic process that we call a z-domain process.

\begin{definition}
    \label{def:z_domain_process}
    Let $(\Xi, F, \mathbb{P})$ denote a probability space.
    A z-domain stochastic process with domain $D\subseteq \C$ is a
    measurable function $f:\Xi\times D\to\C$.
\end{definition}

Note that each value of $\xi\in\Xi$ yields a function
$f_\xi=f(\xi,\cdot):D\to\C$, which is called either a ``realization'' or a
``sample path'' of $f$.
If we take $\xi$ to be selected at random according to the probability law
$\mathbb{P}$, then $f_\xi$ represents a ``random function'' in the frequentist
sense.
Alternatively, if we have a prior belief
about the likelihood of some $f_\xi$ over others,
we may encode this belief in a Bayesian sense using the measure $\mathbb{P}$.
We drop the dependence of $f$ on $\xi$ from the notation outside of definitions,
as it will be clear when $f(z)$ refers to the random variable $f(\cdot,z)$ or
when $f$ stands for a realization $f_\xi$.

\begin{definition}
    \label{def:gaussian_process}
    A z-domain Gaussian process is a z-domain process $f$ such that, for any
    $n$, the random vector
    $(f(z_1),\dotsc,f(z_n))$ is complex multivariate Gaussian-distributed
    for all $(z_1,\dotsc,z_n)\in D^n$.
\end{definition}
Analogous to the way that a real Gaussian process is determined by its mean and
covariance, a z-domain Gaussian process $f$ is completely specified by its mean
$m:D\to\C$,
Hermitian covariance $k:D\times D\to\C$,
and complementary covariance $\tilde{k}:D\times D\to \C$,
defined as
\begin{equation}
    \begin{aligned}
        m(z) = \Ex{f(z)},\quad
        k(z,w) &= \Ex{(f(z)-m(z))(f(w)-m(w))^*},\\
        \tilde{k}(z,w) &= \Ex{(f(z)-m(z))(f(w)-m(w))}.
    \end{aligned}
\end{equation}
In the case that $\tilde{k}(z,w)=0$ for all $z,w\in D$, the real and imaginary
parts of $f$ are independent and identical processes; in this case $f$ is called
a \emph{proper} z-domain Gaussian process.

\subsection{$H_\infty$ Gaussian Processes}%
\label{sub:z_domain_and_h_infty_gaussian_processes}

Consider a deterministic input-output operator $H_g$ with transfer function
function $g:D\to\C$. The condition that $H_g$ belong to the operator space
$H^\infty$ of LTI, causal, and BIBO stable systems is that $g$ belong to the
function space $H_\infty$. Now suppose we wish to construct
a random operator $H_f$ using the realizations of a z-domain process $f$ as its
transfer function: the analogous condition is that the realizations of $f$
lie in $H_\infty$ with probability one.

\begin{definition}
    \label{def:hinf_process}
    A z-domain process is called an $H_\infty$ process when the set
    $\{\xi\in\Xi : f_\xi \in H_\infty\}$ has measure one under $\mathbb{P}$.
\end{definition}

Less formally, an $H_\infty$ process is a z-domain process $f$ such that
$\mathbb{P}(f\in H_\infty)=1$. Having $f_\xi\in H_\infty$ implies that
$\bar{E}\subseteq D$: we usually take $D=\bar{E}$.
If we also require that $H_g$ give real
outputs to real inputs in the time domain, $g$ must satisfy the
conjugate symmetry relation $g(z^*)=g^*(z)$ for all $z\in D$. The analogous
condition for $H_f$ is to require that $f$ satisfy the condition with
probability one.

\begin{definition}
    \label{def:conjugate_symmetric}
    A z-domain process $f$ is called conjugate symmetric when the set
    $\{\xi\in\Xi : f_\xi(z^*)=f_\xi^*(z),\ \forall z\in D\}$ has measure one under $\mathbb{P}$.
\end{definition}

Combining definitions~\ref{def:gaussian_process},~\ref{def:hinf_process},
and~\ref{def:conjugate_symmetric}, we
arrive at our main object of study: conjugate-symmetric $H_\infty$ Gaussian
processes.

\begin{example}[``Cozine'' process]
    \label{ex:cozine}
    Consider the random transfer function
    \begin{equation}
        \label{eq:cozine_form}
        f(z) = \frac{X - a(X\cos(\omega_0) - Y\sin(\omega_0))z^{-1}}
                    {1-2a\cos(\omega_0)z^{-1} + a^2z^{-2}},
    \end{equation}
    where $X,Y\simiid\Norm{0,1}$, $a\in(0,1)$, $\omega_0\in[0,\pi]$. Then $f$
    is a z-domain Gaussian process. From the form of the transfer function, we
    see that $f$ is bounded on the unit circle, analytic on $E$, and conjugate
    symmetric with probability one, from which it follows that $f$ is
    a conjugate symmetric $H_\infty$ process.
    Since $f$ corresponds to the z-transform of an exponentially decaying
    discrete cosine with random magnitude and phase, we call it a
    \emph{``cozine'' process}.
    The process has mean zero, and
    its Hermitian and complementary covariances are
    \begin{equation}
        \begin{aligned}
        \label{eq:cozine_covariances}
            k(z,w)
            &=
            \frac{1-a\cos(\omega_0)(z^{-1}+(w^*)^{-1}) + a^2(zw^*)^{-1}}
            {(1-2a\cos(\omega_0)z^{-1} + a^2z^{-2})(1-2a\cos(\omega_0)(w^*)^{-1} + a^2(w^*)^{-2})},
            \\
            \tilde{k}(z,w)
            &=
            \frac{1-a\cos(\omega_0)(z^{-1}+w^{-1}) + a^2(zw)^{-1}}
            {(1-2a\cos(\omega_0)z^{-1} + a^2z^{-2})(1-2a\cos(\omega_0)w^{-1} + a^2w^{-2})}.
        \end{aligned}
    \end{equation}
\end{example}
As a Bayesian prior for an $H^\infty$ system, this process represents a belief
that the transfer function exhibits a
resonance peak (of unknown magnitude) at $\omega_0$. Knowing $\omega_0$ in
advance is a strong belief, but it can be relaxed by taking a hierarchical model
where $\omega_0$ enters as a hyperparameter. When used as a prior, the
hierarchical model
represents the less determinate belief that there is a resonance peak
\emph{somewhere}, whose magnitude can be made arbitrarily small if no peak is
evident in the data.

The construction in Example~\ref{ex:cozine}, where properties of conjugate
symmetry and BIBO stability can be checked directly, may be extended to random
transfer functions of any finite order. However, the technique does not carry to
the infinite-order $H_\infty$ processes required for nonparametric Bayesian
system identification, or more generally for applications that do not place an \emph{a
priori} restriction on the order of the system.
We are therefore motivated to find conditions under which a z-domain process is
a conjugate-symmetric $H_\infty$ Gaussian process expressed directly in
terms of covariance properties.

\section{Constructing $H_\infty$ Gaussian Processes}
\label{sec:hinf-processes}

This section investigates two tools
to verify that a z-domain Gaussian process has the $H_\infty$
property with conjugate symmetry. The first tool is a setof
sufficient conditions on the mean and covariances that, when satisfied by the
process, ensure the $H_\infty$ property and conjugate symmetry. These
conditions are developed in Section~\ref{sub:general_sufficient_conditions} and
summarized in Theorem~\ref{thm:hinf-gp-conditions}.
The second tool is a special class of symmetric z-domain processes, provided in
Theorem~\ref{prop:stationary-process} guaranteed to have the $H_\infty$
property. 
This class of processes, which enjoys the property of \emph{Hermitian
Stationarity}, is
parameterized by nonnegative real $\ell_1$ sequences. One
may establish that a given Hermitian stationary process is $H_\infty$ by finding the
corresponding sequence; conversely, one may select an arbitrary $\ell_1$
sequence and receive a Hermitian stationary $H_\infty$ process.
This class of processes is developed in
Section~\ref{sub:hermitian_stationary_processes}, and the $\ell_1$
characterization is given in Theorem~\ref{prop:stationary-process}.

Throughout this section, we consider a z-domain Gaussian process $f$ with zero mean,
Hermitian covariance function $k$, and complementary covariance function
$\tilde{k}$. Taking zero mean implies no loss in generality: to lift any of
these conditions to a process with nonzero mean, we simply ask that the desired
property (inhabiting $H_\infty$, possessing conjugate symmetry, or both) also
hold for the mean.

\subsection{General Sufficient Conditions}%
\label{sub:general_sufficient_conditions}

For a general z-domain Gaussian process, we can establish almost sure
boundedness and analyticity (and hence the $H_\infty$ property) with corresponding
regularity conditions on the covariance functions. The essential condition is
that $k(z,w)$ possess derivatives with respect to $z$ and $w^*$. The
following lemma, whose proof is deferred to Section~\ref{sec:proofs_of_main_results}, demonstrates the case
of proper z-domain GPs on compact domains.

\begin{lemma}
    \label{lem:analytic-process-compact}
    Suppose $k:D\times D\to\C$ is positive definite and bounded for a simply
    connected domain $D\subseteq\C$ that is closed under
    conjugation,\footnote{Here, ``positive definite'' is meant in the
    kernel-theoretic sense, that is that the kernel Gramian matrix
$(K)_{ij}=k(z_i,z_j)$ is positive definite for any set $z_1,\dotsc,z_n\in D$.}
that is $z\in D\Rightarrow z^*\in D$. 
    Furthermore, suppose that $k$ is holomorphic in its first argument and
    \emph{antiholomorphic} in its second argument; that is, $k(z,w)$ possesses
    complex derivatives of all orders with respect to $z$ and $w^*$.
    Let 
    $k_{11}=\frac{\partial^2}{\partial z \partial w^*} k$
    and
    suppose that there exists $\alpha > 0$ such that
    \begin{equation}
        k_{11}(z,z) + k_{11}(w,w) - 2k_{11}(z,w) \le |z-w|^{2+\alpha}.
    \end{equation}
    Then there exists a proper z-domain Gaussian process with Hermitian
    covariance $k$ whose realizations are analytic with probability one.
\end{lemma}

We can adapt Lemma~\ref{lem:analytic-process-compact} to the case of
conjugate-symmetric processes on the exterior of the unit disk using the
following constructions. First, we can symmetrize an arbitrary proper process to
yield a conjugate-symmetric process. If the Hermitian covariance of the proper
process satisfies a symmetry condition, then it is preserved under the
symmetrization.

\begin{lemma}
    \label{lem:symmetric-from-proper}
    Suppose $f$ is a mean-zero proper process whose Hermitian covariance $k$
    satisfies
    $
        k(z^*,w^*)^* = k(z,w)
    $
    on a domain $D$ closed under conjugation.
    Then the process $g$ defined pointwise as
    $g(z)=\frac{1}{\sqrt{2}}(f(z) + f(z^*)^*)$
    is mean-zero, conjugate-symmetric z-domain Gaussian process on the domain
    $D$ with Hermitian covariance $k$.
    Furthermore, if $f$ has almost surely analytic realizations, then so does $g$.
\end{lemma}

\begin{proof}
    That $g$ is conjugate-symmetric is evident from its construction,
    since
    \begin{equation*}
        g(z)^* = \tfrac{1}{\sqrt{2}}(f(z) + f(z^*)^*)^*
    = \tfrac{1}{\sqrt{2}}(f(z)^* + f(z^*)) = g(z^*).
    \end{equation*}
    Furthermore,
    if $f$ is analytic then so is
    $f(z^*)^*$, and so thereby is $g$.
    Finally, we have
    \begin{equation*}
        \begin{aligned}
            \Ex{g(z)g(w)^*}
            &=
            \tfrac{1}{2}\Ex{\left(f(z) + f(z^*)^*\right)\left(f(w)+f(w^*)^*\right)^*}\\
            &= \tfrac{1}{2} \Ex{f(z)f(w)^* + f(z)f(w^*) + f(z^*)^*f(w)^* + f(z^*)^*f(w^*)}\\
            &= \tfrac{1}{2}\left( k(z,w) + \tilde{k}(z,w^*) + \tilde{k}(z^*,w)^* +
            k(z^*,w^*)^*\right)
            = k(z,w)
        \end{aligned}
    \end{equation*}
    where the last equality follows by our symmetry assumption; the
    complementary covariance terms vanish because $f$ is proper.
    This confirms that the Hermitian covariance of $g$ is $k$.
\end{proof}

Furthermore, we can transform an almost surely bounded and analytic process on the interior
of the unit disk into a bounded and analytic process on the exterior of the unit
disk-- that is, an $H_\infty$ process-- through the following elementary conformal mapping.

\begin{lemma}
    \label{lem:exterior-from-interior}
    Suppose $f(z)$ is analytic for $|z|\le 1$. Then $g(z)=f(z^{-1})$ is analytic
    for $|z|\ge 1$.
\end{lemma}

Assembling
lemmas~\ref{lem:analytic-process-compact},~\ref{lem:symmetric-from-proper},
and~\ref{lem:exterior-from-interior}, we arrive at the following sufficient
conditions for a z-domain Gaussian process to be an $H_\infty$ process.

\begin{theorem}
    \label{thm:hinf-gp-conditions}
    Suppose $k: D\times D\to\C$ is a positive definite kernel function that is
    bounded and analytic on a compact domain $D$ whose interior contains the
    unit circle and its interior, satisfying the assumptions of
    Lemmas~\ref{lem:analytic-process-compact}
    and~\ref{lem:symmetric-from-proper}. Then there exists a conjugate-symmetric
    z-domain Gaussian process whose realizations are bounded and analytic on the
    exterior of the unit disk, whose Hermitian covariance is given by $k(z^{-1},
    w^{-1})$.
\end{theorem}

\begin{proof}
    By Lemmas~\ref{lem:analytic-process-compact}
    and~\ref{lem:symmetric-from-proper},
    we know that there exists a mean-zero, conjugate-symmetric process $f$ whose
    domain is $D$ and whose realizations are bounded and analytic in that
    domain. Let $g$ denote the z-domain process on the unit circle and
    its exterior whose realizations are defined pointwise as $g_\xi(z) = f_\xi(z^{-1})$.
    Since $f_\xi$ is bounded and analytic for $|z|\le 1$, it follows from
    Lemma~\ref{lem:exterior-from-interior} that $g_\xi$ is bounded and analytic
    for $|z|\ge 1$, that is for the unit circle and its exterior. Furthermore,
    conjugate symmetry is also preserved. Finally
    we have
    \begin{equation*}
        \Ex{g(z)g(w)^*} = \Ex{f(z^{-1})f(w^{-1})^*} = k(z^{-1},w^{-1})
    \end{equation*}
    which establishes the Hermitian covariance.
\end{proof}

\subsection{Hermitian Stationary Processes}%
\label{sub:hermitian_stationary_processes}

We now turn to a special class of $H_\infty$ Gaussian processes whose
covariances admit a direct characterization in terms of positive real $\ell_1$
sequences. Following Theorem~\ref{thm:hinf-gp-conditions}, we first restrict our
attention to Hermitian
covariance functions $k(z,w)$ holomorphic in $z$ and antiholomorphic in $w$ in a region
containing the exterior of the unit disk. Such functions admits the double
Laurent expansion
    $k(z,w) = \sum_{n,m=0}^\infty a_{nm} z^{-n}(w^*)^{-m}$;
we develop our special class by restricting off-diagonal $a_{nm}$ to be zero and
renaming $a_{nn}=a_n^2$,
yielding 
    $k(z,w) = \sum_{n=0}^\infty a_{n}^2 (z w^*)^{-n}$
where $a_n$ is real and positive. When $z$ and $w$ are restricted to the unit
circle, the resulting covariance,
    $k(\ej{\theta},\ej{\phi}) = \sum_{n,m=0}^\infty a_{n}^2 \ej{n(\theta-\phi)}$,
reduces to a function of the difference between the arguments of the inputs.
This condition is similar to the condition obeyed by real-valued stationary
processes, so we call a z-domain process with this property a \emph{Hermitian
stationary} process.
\begin{definition}
    A z-domain Gaussian process is \emph{Hermitian stationary} if its
    Hermitian covariance satisfies
    $k(e^{j\theta}, e^{j\phi})=k(e^{j(\theta-\phi)},1)$ for all $\theta$,$\phi\in[-\pi,\pi)$.
\end{definition}
Using a stationary process as a prior is common practice in
machine learning and control-theoretic applications of Gaussian process models.
Stationary processes are useful for constructing regression priors that do not
introduce unintended biases in their belief about the frequency response: since
$f(e^{j\theta})$ has the same
Hermitian variance across the entire unit circle, a sample path from a Hermitian
stationary $H_\infty$ process is just as likely to exhibit low-pass behavior as
it is high-pass or band-pass.%
We can obtain a ``partially informative'' prior by adding an $H_\infty$ process encoding
strong beliefs in one frequency range (such as the presence of a resonance peak)
to an $H_\infty$ process encoding weaker beliefs across all frequencies. The
sum, also an $H_\infty$ process, encodes a combination of these beliefs.

Under the additional condition of Hermitian stationarity, the
$H_\infty$ process is characterized by a sequence of nonnegative constants. This
is demonstrated in by the following result, whose proof is deferred to
Section~\ref{sec:proofs_of_main_results}.

\begin{theorem}
    \label{prop:stationary-process}
    Let $f$ be a Hermitian stationary, conjugate-symmetric z-domain Gaussian
    process with continuous Hermitian covariance $k$ and complementary
    covariance $\tilde{k}$.
    Then $f$ is an $H_\infty$ process if and only if $k$ and $\tilde{k}$ have
    the form
    \begin{equation}
        \label{eq:covariance_expansion}
        k(z,w) = \sum_{n=0}^\infty a_n^2 (zw^*)^{-n},
        \qquad
        \tilde{k}(z,w) = \sum_{n=0}^\infty a_n^2 (zw)^{-n},
    \end{equation}
    where $\{a_n\}_{n=0}^\infty$ is a nonnegative real $\ell^1$ sequence.
    Furthermore, $f$ may be expanded as \begin{equation}
        \label{eq:process_expansion}
        f(z) = \sum_{n=0}^\infty a_n w_n z^{-n},
    \end{equation}
    where $w_n\simiid\Norm{0,1}$.
\end{theorem}

Theorem~\ref{prop:stationary-process} provides a useful tool for
constructing conjugate-symmetric $H_\infty$ Gaussian processes: all we need to
do is select a summable sequence of nonnegative numbers.

\begin{example}[Geometric $H_\infty$ process]
    \label{ex:geometric}
    Take $a_n^2 = \alpha^n$ with $\alpha\in(0,1)$; this yields
    a conjugate-symmetric $H_\infty$ Gaussian process with Hermitian covariance
    $k_\alpha(z,w)=\sum_{n=0}^\infty \alpha^n (zw^*)^{-n} = \frac{zw^*}{zw^*-\alpha}$
    and complementary covariance
    $\tilde{k}_\alpha(z,w) = \frac{zw}{zw-\alpha}$.
\end{example}

\section{Gaussian Process Regression in the Frequency Domain}
\label{sec:regression}

Let $H_\Delta\in H^\infty$ denote a system uncertainty whose transfer function $\Delta\in H_\infty$
we wish to identify. While not necessarily stochastic, $\Delta$ is unknown, and we represent both
our uncertainty and our prior beliefs in a Bayesian fashion with an $H_\infty$ Gaussian process
with
Hermitian and complementary covariances $k$ and $\tilde{k}$. To model our prior
beliefs, the distribution of $\Delta$ should give greater probability to functions we
believe are likely to correspond to the truth, and should assign probability zero to
functions ruled out by our prior beliefs. As an example of the latter, the fact that
$P(\Delta\in H_\infty)=1$ encodes our belief that $\Delta\in H_\infty$, which demonstrates
the importance of $H_\infty$ Gaussian processes for prior model design.

We suppose that our data consists of $n$ noisy frequency-domain point estimates
$y_i = \Delta(z_i) + e_i$, where $e_i\simiid\Norm{0,\sigma_n^2}$, $z_i\in\bar{E}$.
If our primary form of data is a time-domain trace of input and output values,
we first convert this data into an \emph{empirical transfer function estimate}
(ETFE). There are several well-established methods to construct ETFEs from time
traces~\cite{hayes1996statistical}, such as Blackman-Tukey spectral analysis,
windowed filter banks, or simply dividing the DFT of the output trace by the DFT
of the input trace. In our numerical examples, we will use windowed filter
banks.

Our approach is essentially the same procedure as standard Gaussian process
regression as described in~\cite{gpml} extended to the complex case.
We take the mean of the prior model to be zero without loss of generality.
To estimate the transfer function at a new point $z$, we note that $\Delta(z)$ is
related to $(y_1,\dotsc,y_n)$ under the prior model as
\begin{equation}
    \begin{bmatrix}
        \Delta(z)
        \\
        y
    \end{bmatrix}
    \sim
    \CNorm{0,
        \begin{bmatrix}
            K_{xx} & K_{xy} \\
            K_{xy}^* & K_{yy} \\
        \end{bmatrix},
        \begin{bmatrix}
            \tilde{K}_{xx} & \tilde{K}_{xy} \\
            \tilde{K}_{xy}^* & \tilde{K}_{yy} \\
        \end{bmatrix}
    };
\end{equation}
where $y\in\C^n$, $K_{yy}\in\C^{n\times n}$,
$K_{xy}\in\C^{n\times 1}$, and
$K_{xx}\in\C$ are defined componentwise as
\begin{equation}
    (y)_i=y_i,
    \quad
    \left(K_{yy}\right)_{ij} = k(z_i,z_j) + \sigma_n^2\delta_{ij},
    \quad
    \left(K_{xy}\right)_{ij} = k(z,z_i),
    \quad
    K_{xx} = k(z,z) + \sigma_n^2,
\end{equation}
and the components of the complementary covariance matrix are defined
analogously.

By conditioning $\Delta(z)$ on the data $y$ according to the prior
model, we obtain the posterior distribution of $\Delta(z)$. According to the
conditioning law for multivariate complex Gaussian random
variables~\cite[\S 2.3.2]{schreier2010statistical}, this is
$\Delta(z)|y\sim\CNorm{\mu, \sigma_p^2, \tilde{\sigma}_p^2}$, where
\begin{equation}
    \begin{aligned}
        \label{eq:widely_linear_prediction}
        \mu_q
        &=
        (K_{xy} - \tilde{K}_{xy}(K^*_{yy})^{-1}\tilde{K}_{yy}^*)P^{-1}y
        +
        (\tilde{K}_{xy}-K_{xy}K_{yy}^{-1}\tilde{K}_{yy})(P^*)^{-1}y^*\\
        \sigma_p^2
        &= k_{zz}
        - K_{xy}P^{-1}K_{xy}^*
        + \tilde{K}_{xy}K_{yy}^{-1}\tilde{K}_{yy}(P^*)^{-1}K_{xy}^*\\
        &\quad
        - \tilde{K}_{xy}(P^*)^{-1}\tilde{K}_{xy}^*
        + K_{xy}K_{yy}^{-1}\tilde{K}_{yy}(P^*)^{-1}\tilde{K}_{xy}^*\\
        \tilde{\sigma}_p^2
        &= \tilde{k}_{zz}
        - K_{xy}P^{-1}(\tilde{K}_{xy}^*)^*
        + \tilde{K}_{xy}K_{yy}^{-1}\tilde{K}_{yy}(P^*)^{-1}(\tilde{K}_{xy}^*)^*\\
        &\quad
        - \tilde{K}_{xy}(P^*)^{-1}(K_{xy}^*)^*
        + K_{xy}K_{yy}^{-1}\tilde{K}_{yy}(P^*)^{-1}(K_{xy}^*)^*
    \end{aligned}
\end{equation}
and where $P$ denotes the Schur complement
$P=K_{yy} - \tilde{K}_{yy}(K_{yy}^*)^{-1}\tilde{K}^*_{yy}$.
The predictive mean $\mu_p$ is the \emph{minimum mean-square error widely linear
estimator} of $\Delta(z)$ given $y$, where ``widely linear'' means
that $\mu_p$ is a linear combination of both $y$ and $y^*$.
A \emph{strictly linear} estimator, on the other hand, uses only $y$. Under the
same circumstances as above, the minimum least-square strictly linear estimator
for $\Delta(z)$ given $y$ and its error variance are respectively
\begin{equation}
    \label{eq:strictly_linear_prediction}
    \hat{\Delta}(z) = K_{xy}^* K_{yy}^{-1} y,
    \qquad
    \sigma^2_\Delta(z) = K_{zz} - K_{xy}^* K_{yy}^{-1}k_{xy},
\end{equation}
where the superscript $^*$ denotes Hermitian transpose.
This mean and variance are identical to the posterior mean and variance of a real Gaussian
process regression model (cf. Equation (2.19) in~\cite{gpml}) except that $K_{xx}$, $K_{xy}$, and
$K_{yy}$ are complex-valued.

The widely linear estimator can only be an improvement on the linear estimator,
since an estimate made using $y$ can certainly be made using
$(y,y^*)$. The improvement is measured by the Schur
complement $P$ defined above, which is the error covariance of statistically
estimating $y^*$ from $y$, or equivalently estimating the real part given the
imaginary part.
In particular, when $P=0$, the strictly linear and widely linear
estimators coincide, and the expressions
in~\eqref{eq:widely_linear_prediction} become ill-defined.
One case where this holds is when the covariances are \emph{maximally improper},
in which case the imaginary part can be estimated from the real with zero error.

In our experiments with real-impulse $H_\infty$ processes, we have found that
$P$ tends to be close to singular, and small in induced 2-norm and Frobenius norm
relative to $K_{yy}$ and $\tilde{K}_{yy}$. This makes the mean and variance
computations in~\eqref{eq:widely_linear_prediction} numerically unstable while
also implying that the strictly linear estimator will perform similarly to the
widely linear estimator.
We believe this is due to the symmetry
condition imposed on $k$ and $\tilde{k}$ by having real impulse response.
This condition implies that the imaginary part can be computed exactly from the
real part by the discrete Hilbert transform~\cite[\S 2.26]{rabiner1975theory}.
The covariance matrices $K_{yy}$ and $\tilde{K}_{yy}$ will not themselves be
maximally improper, since the Hilbert transform requires knowledge over the
entire unit circle; however, our experiments suggest that they are close to
maximally improper, and we conjecture that they become maximally improper in the
limit of infinite data.
This suggests that the strictly linear estimator will perform well for
conjugate-symmetric $H_\infty$ priors. For this reason, as well as the numerical
instability of the widely linear estimator when $P$ is close to singular, we use
the strictly linear estimator in our numerical experiments.

For $z\in D$ and $\eta > 0$, define the \emph{confidence ellipsoid}
$\mathcal{E}_\eta(z) = \{w\in \C: |w-\hat{\Delta}(z)|^2 \le \eta^2\sigma^{2}_\Delta(z)\}$. By
Markov's inequality, we know that $\Delta(z)\in\mathcal{E}_\eta(z)$ with probability
$\ge 1 - 1/\eta^2$.
This implies bounds on the real and imaginary parts
by projecting the confidence ellipsoid onto the real and imaginary axes:
from these we can construct probabilistic bounds on the magnitude
and phase of $\Delta(z)$ via interval arithmetic, which we will see in the numerical examples.

Let $\theta\in\Theta$ denote the hyperparameters of a covariance
function $k_\theta$, so that $K_{yy}$ becomes a function of $\theta$:
then the log marginal likelihood of the data under the posterior for the strictly
linear case is
$
    L(\theta) = -\tfrac{1}{2}\left(
        y^*K_{yy}(\theta)^{-1}y + \log\det K_{yy}(\theta) + n\log 2\pi
    \right).
$
Keeping the data $y$ and input locations $z_i$ fixed, $L(\theta)$ measures the
probability of observing data $y$ when the prior covariance function is
$k_\theta$. By maximizing $L$ with respect to $\theta$, we find the covariance
among $k_\theta$, $\theta\in\Theta$ that best explains the observations.%
\footnote{Although it seems contradictory to choose prior parameters based on
posterior data, it can be justified as an empirical-Bayes approximation to a
hierarchical model with $\theta$ as hyperparameter. }

\begin{figure}[htpb]
    \centering
    \includegraphics[width=\linewidth]{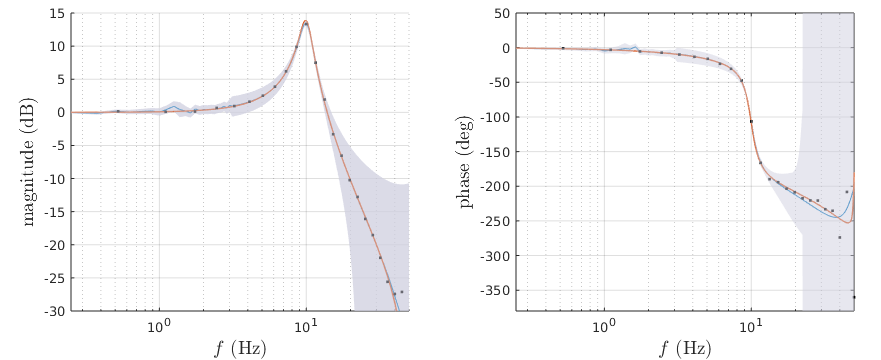}
    \caption{Bode plot of the second-order resonant system (orange), and its estimate
    (blue) using $H_\infty$ Gaussian process regression from an empirical
transfer function estimate (black points) with $\eta=3$ confidence ellipsoid bounds (grey).}%
    \label{fig:resonance_regression}
\end{figure}

We now demonstrate the process by applying the strictly-linear $H_\infty$ Gaussian process regression method
to the problem of identifying a second-order system that exhibits a resonance peak.
The system is specified in continuous time, with canonical second-order transfer
function
\begin{equation}
    g(s) = \frac{\omega_0^2}{s^2+2\xi\omega_0s + \omega_0^2},
\end{equation}
where $\omega_0=20\pi$ rad/s, and $\xi=0.1$, and converted to the discrete-time
transfer function $g(z)$ using a zero-order hold discretization with a sampling
frequency of $f_s=100$ Hz.
We suppose that we know \emph{a priori} that there is a resonance peak, but not
about its location or half-width, and we have no other strong information about
the frequency response. For this prior belief, an appropriate prior model is a
weighted mixture of a cozine process and a Hermitian stationary process.
In particular, we use the family of $H_\infty$ processes with covariance
functions
\begin{equation}
    k(z,w) = \sigma^2_g k_g(z,w) + \sigma^2_c k_c(z,w),
    \tilde{k}(z,w) = \sigma^2_g \tilde{k}_g(z,w) + \sigma^2_c \tilde{k}_c(z,w),
\end{equation}
where $k_g$ is the covariance of the geometric $H_\infty$ process defined in
Example~\ref{ex:geometric}, and $k_c$ is the covariance of the cozine process,
and likewise for the complementary covariance. $\sigma^2_g$ and $\sigma^2_c$ are
weights that determine the relative importance of the two
parts of the model. This family of covariances has five hyperparameters:
$\sigma_g\in[0,\infty)$,
$\alpha\in(0,1)$,
$\sigma_c\in[0,\infty)$,
$\omega_0\in[0,\pi]$, and
$a\in(0,1)$.

We suppose that an input trace $u(n)$ of Gaussian white noise with variance $\sigma^2_u=1/f_s$ is run through
$H_g$ yielding an output trace $y(n)$; our observations comprise these two
traces, corrupted by additive Gaussian white noise of variance
$\sigma^2=10^{-4}/f_s$. To obtain an empirical transfer function estimate,
we run both observation traces through a bank of 25 windowed 1000-tap DFT filters.
The impulse responses of the filter bank are
$h_i(n) = e^{j\omega_i n} w(n)$ for $i=1,\dotsc,25$,
with Gaussian window $w(n)=\exp(-\tfrac{1}{2}(\sigma_w(n-500) / 1000)^2)$ for
$n=0,\dotsc,999$, and $w(n)=0$ otherwise, with window half-width
$\sigma_w=0.25$.
Let $u_i$, $y_i$ denote the outputs of filter $h_i$ with inputs $u$, $y$
respectively: $y_i(n)/u_i(n)$ gives a running estimate of $g(e^{j\omega_i})$,
whose value after 1000 time steps we take as our observation at
$z_i=e^{j\omega_i}$.
Figure~\ref{fig:resonance_regression} shows the regression from the strictly
linear estimator~\eqref{eq:strictly_linear_prediction} after tuning the
covariance hyperparameters via maximum likelihood, along with predictive error
bounds based on $\eta=3$ confidence ellipsoids.

\section{Robustness Analysis with $H_\infty$ Gaussian Processes}
\label{sec:hinf-gp-robustness}

Now that we have explored how to refine an $H_\infty$ GP model from online data,
we move to a second and equally important problem: how to establish
probabilistic robustness guarantees for $H_\infty$ processes. Our goal is to
establish probabilistic guarantees of robustness; in other words,
establishing that a given certificate of robustness is obtained with probability
$\ge 1-\epsilon$ for a prescribed $\epsilon\in(0,1)$. For several classical
robustness certificates, this reduces to the problem of bounding
the \emph{gain excursion probability}
\begin{equation*}
    P_\gamma(f) = \Pr{\sup_{\Omega\in[-\pi,\pi)} |f(\ejo)| > \gamma}
\end{equation*}
of a general $H_\infty$ GP $f$, which
measures how likely the
$L_2$ gain of $H_f$ is to exceed the level $\gamma$. 
A simple example of how excursion problems arise is the problem of extending small-gain
arguments to the probabilistic case with an $H_\infty$ GP uncertainty.
Consider a
nominal plant $G$ in feedback with an uncertainty $\Delta$ modeled by an
$H_\infty$ GP. 
Assume the plant satisfies 
$\norm{G}_\infty < M$; then
it follows from the small-gain condition~\cite[Theorem
III.2.1]{desoer2009feedback} 
that the interconnection is stable if $\norm{\Delta}_\infty \le \frac{1}{M} \eqdef \gamma$. 
Thus bounding $P_{\gamma}(\Delta) \le \epsilon$
amounts to proving that the interconnection
is stable for an ensemble of realizations with probability $\ge 1-\epsilon$.

A more general example is the problem of proving an probabilistic
guarantee that an $H_\infty$ GP satisfies an integral quadratic constraint
(IQC). 
In the
discrete time SISO setting, an IQC with multiplier $\Pi$ is a behavioral
constraint on signals $(v,w)\in H_2\times H_2$
of the form~\cite{hu2017robustness,megretski1997system,veenman2016robust}
\begin{equation}
    \label{eq:iqc-def}
    \int_{-\pi}^\pi
    \begin{bmatrix}
        v(\ejo) \\ w(\ejo)
    \end{bmatrix}^*
    \Pi(\ejo)
    \begin{bmatrix}
        v(\ejo) \\ w(\ejo)
    \end{bmatrix}
    d\Omega
     \ge 0
\end{equation}
where $\Pi:[-\pi,\pi)\to\C^{2\times 2}$.
A conjugate-symmetric system $\Delta: H_2\to H_2$ is said to satisfy the IQC with a
multiplier $\Pi$ if~\eqref{eq:iqc-def} is satisfied by pairs $(v,\Delta v)$ for
all conjugate-symmetric $v\in H_2$.
IQCs are able to express a wide range of behavioral properties, and knowing that
an uncertainty satisfies a particular IQC is a powerful tool for constructing
controllers that are robust against uncertainties satisfying that IQC.

If $\Delta$ is an LTI operator, then under fairly mild conditions, the IQC has a
simple geometric interpretation.  Specifically, the if $\Delta$
satisfies~\eqref{eq:iqc-def}
then its frequency response $\Delta(\ejo)$ is constrained, pointwise in
frequency, to lie: (a) within a circle if $\Pi_{22}(\ejo)<0$; (b)
outside a circle if $\Pi_{22}(\ejo)>0$; or (c) on one side of a half
space if  $\Pi_{22}(\ejo)=0$.  We shall state this result formally
for case (a), but an analogous statement can be given for the other two
cases. The result is provided under the assumption
$\Pi_{22}(\ejo)=-1$. This normalization is without loss of generality
because we can scale the multiplier $\Pi$ by a positive constant.

\begin{lemma}[adapted from~\cite{pfifer2015integral}, Lemma 1 (i)]
        Suppose that an LTI system with conjugate-symmetric transfer function $\Delta$ satisfies an
        IQC with continuous, conjugate-symmetric multiplier $\Pi$ normalized to
        $\Pi_{22}(\ejo) = -1$. Then for
        each $\Omega\in[-\pi,\pi)$, $\Delta(\ejo)$ lies in a circle in the Nyquist
        plane with center
        $\Pi_{21}(\ejo)$ and radius 
 
        \noindent $\sqrt{|\Pi_{11}(\ejo)|^2 + |\Pi_{21}(\ejo)|^2}$.
\end{lemma}
This condition is evidently equivalent to the condition that
\begin{equation}
    \frac{|\Delta(\ejo) - \Pi_{21}(\ejo)|}{\sqrt{|\Pi_{11}(\ejo)|^2 + |\Pi_{21}(\ejo)|^2}} \le 1\ \forall\Omega\in[-\pi,\pi).
\end{equation}
Thus the problem of
establishing that an $H_\infty$ GP satisfies an IQC with probability $\ge 1-\epsilon$
reduces to a gain excursion probability problem, namely
proving that
\begin{equation}
    \Pr{
        \sup_{\Omega\in[-\pi,\pi)} \left|\frac{\Delta(\ejo) - \Pi_{21}(\ejo)}{\sqrt{|\Pi_{11}(\ejo)|^2 + |\Pi_{21}(\ejo)|^2}}\right| > 1
    } \le \epsilon,
\end{equation}
or in other words that
\begin{equation}
    \label{eq:iqc-gp-condition}
    P_1\left(
        \frac{\Delta(\ejo) - \Pi_{21}(\ejo)}{\sqrt{\Pi_{11}(\ejo) + |\Pi_{21}(\ejo)|^2}}
    \right) \le \epsilon.
\end{equation}

\subsection{Bounding the excursion probability}

Having established how gain excursion probabilities arise in proving
probabilistic safety guarantees for $H_\infty$ uncertainties, we turn to the
problem of how to bound these probabilities.
It is not generally possible to directly compute $P_\gamma(f)$; however, we can bound
it from above using a related quantity, the expected number of gain upcrossings.

Associated to any
$H_\infty$ Gaussian process is its
\emph{gain process} $ |f(\ejo)|, \Omega\in[-\pi,\pi)$. Assuming
that the gain process is differentiable with respect to $\Omega$, a
\emph{gain upcrossing} at level $\gamma$ is a value $\Omega_c$ such that
$|f(\ej{\Omega_c})|=\gamma$ and $\tfrac{\partial}{\partial\Omega}
|f(\ej{\Omega_c})|>0$.
Under the assumptions given so far, there can be at most finitely many
upcrossings, so the random variable 
$N_\gamma = \#\{\Omega\in[-\pi,\pi) : |f(\ejo)|=\gamma\text{ and } \tfrac{\partial}{\partial\Omega} |f(\ejo)|>0\}$,
where $\#S$ denotes the cardinality of a set $S$,
is well-defined and its
expectation $\Ex{N_\gamma}$ is almost surely finite.
A simple application of Markov's
inequality yields the bound
\begin{equation}
    \label{eq:excursion-bound}
    \begin{aligned}
        P_\gamma(f)
        &= \Pr{|f(\ej{0})| > \gamma} + \Pr{|f(\ej{0})| \le \gamma, N_\gamma \ge 1}\\
        &\le \Pr{|f(\ej{0})| > \gamma} + \Ex{N_\gamma}.
    \end{aligned}
\end{equation}
The reason that we consider a bound for $P_\gamma$ rather than a direct computation
is that direct computation of $P_\gamma$ is only possible in the simplest cases. On
the other hand, $\Ex{N_\gamma}$ can be computed with a closed-form (though sometimes
complicated) expression as long as the process is differentiable. 

The lack of Gaussian structure in $|f(\ejo)|$ would make it difficult to compute
this formula
directly from $f$ (e.g. by applying a Rice formula like~\cite[Theorem 3.4]{azais2009level}).
To overcome the difficulty, 
we reframe the problem as a \emph{vector crossing} problem on the vector
Gaussian process
$g(\Omega) = (x(\ejo), y(\ejo))$ formed from the real and imaginary parts:
the gain process $|f(\ejo)|$ crosses from $\le \gamma$ to $> \gamma$ precisely when
the vector process $g(\Omega)$ crosses from the interior of
the circle $x^2 + y^2 = \gamma^2$ to the exterior.
By taking this perspective, we relinquish the
topological simplicity of the scalar crossing problem in order to retain the
Gaussian structure of the stochastic process. While we cannot apply Rice
formulas in the vector setting, there are analogous results for counting vector
crossings. We use the following result due to Belyaev.

\begin{theorem}[first-order Belyaev formula~\cite{belyaev1968number}] 
    \label{thm:belyaev_formula}
    Let $g:\Xi\times [0,T]\to\R^n$ be a vector-valued stochastic process and
    $\Phi:\R^n\to\R$ a boundary function satisfying the following conditions:
    \begin{enumerate}
        \item $g$ is continuously differentiable with probability one, and the
            random variables $g(t)$, $t\in T$ all possess densities $p_{g(t)}$;
        \item The conditional densities $p(x|y)$ exist for $x=g(t)$,
            $y=g^\prime(t)$, and the densities depend continuously on $x$.
        \item $\Phi$ is continuously differentiable, and to each
            $\epsilon$-neighborhood of the surface 
            $S_\Phi=\{x\in\R^n : \Phi(x)=0\}$ 
            we can associate coordinates
            $(\phi, \zeta_1,\dotsc,\zeta_{n-1})$, where
            $\phi=\inf_{y\in S_\Phi} ||x-y||_2$;
    \end{enumerate}
    Let $N_\Phi$ denote the number of times a realization $g_\xi$ of $f$ exits
    the surface $S_\Phi$: then
    \begin{equation}
        \label{eq:belyaev_formula}
        \Ex{N_\Phi} 
        = \int_0^T\int_{S_\phi} \Ex{n{((x)}^\top g(t))_+ | g'(t)=x} p_{f(t)}(x) ds(x) dt,
    \end{equation}
    where $(x)_+ = \max(0,x)$, and
    $n(x)$ is the outward-facing unit normal vector of $S_\Phi$ at the point $x$.
\end{theorem}

Applying the first-order Belyaev formula to the real and imaginary parts of an $H_\infty$ GP and the surface
$x^2+y^2=\gamma^2$ yields the following formula for the expected number of gain
upcrossings, whose proof is deferred to
Section~\ref{sec:proofs_of_main_results}. 
 
\begin{theorem}
    \label{thm:gain_upcrossing_formula}
    Consider an $H_\infty$ Gaussian process $f$ with mean $m_x + jm_y$ and
    Hermitian and complementary covariances $k$, $\tilde{k}$. Let $N_\gamma$ denote
    the integer-valued random variable that counts the number of $\gamma$-level gain
    upcrossings of $f$. 
    Suppose that $m_x(\ejo)$ and $m_y(\ejo)$ are differentiable with respect to
    $\Omega$ and that 
    $k(\ej{\Omega_1},\ej{\Omega_2})$
    and
    $\tilde{k}(\ej{\Omega_1},\ej{\Omega_2})$
    are thrice differentiable with respect to $\Omega_1$ and $\Omega_2$.
    Then the expected number of gain upcrossings between
    frequencies $-\pi$ and $\pi$ is
    \begin{equation}
        \begin{aligned}
            \label{eq:gain_upcrossing_formula}
            \Ex{N_\gamma}
            &=
            \int_{-\pi}^\pi \int_0^{2\pi}
            \frac{\gamma}{2\pi}
            \Bigg(
                \frac{\sigma_{z(\theta)}(\Omega)}{\sqrt{2\pi}}
                e^{-\frac{1}{2}
                (\mu_{z(\theta)}(\Omega)/\sigma_{z(\theta)}(\Omega))^2}\\
            &\qquad\qquad\qquad+
                \frac{1}{2} \mu_{z(\theta)}(\Omega)
                \left(1+\erf\left(\frac{\mu_{z(\theta)}(\Omega)}{\sqrt{2}
                \sigma_{z(\theta)}(\Omega)}\right)\right)
            \Bigg)\\
            &\times \det\Sigma(\Omega,\Omega)^{-1/2}
                e^{-\frac{1}{2}(z(\theta)-m(\Omega))^\top \Sigma(\Omega,\Omega)^{-1}(z(\theta)-m(\Omega)} d\theta d\Omega,
        \end{aligned}
    \end{equation}
    where
    \begin{align}
        z(\Omega) &= [\gamma \cos(\Omega); \gamma\sin(\Omega)]\\
        \mu_z(\Omega)
        &=\gamma^{-1}z^\top m(\Omega)^\prime + \gamma^{-1}z^\top C(\Omega,\Omega)\Sigma(\Omega,\Omega)^{-1}(z-m(\Omega))\\
        \sigma_z(\Omega)
        &=\gamma^{-2} z^\top(\Sigma^\prime(\Omega,\Omega) - C(\Omega,\Omega)(\Sigma(\Omega,\Omega)^{-1}C(\Omega,\Omega)^\top) z
    \end{align}
    \begin{align}
        \Sigma(\Omega,\Omega)
        &=
        \begin{bmatrix}
            k_x(\Omega,\Omega) & k_c(\Omega,\Omega) \\
            k_c(\Omega,\Omega) & k_y(\Omega,\Omega)
        \end{bmatrix},
        \Sigma^\prime(\Omega,\Omega)
        =
        \begin{bmatrix}
            k_x^{12}(\Omega,\Omega) & k_c^{12}(\Omega,\Omega) \\
            k_c^{12}(\Omega,\Omega) & k_y^{12}(\Omega,\Omega)
        \end{bmatrix},\\
        C(\Omega,\Omega)
        &=
        \begin{bmatrix}
            k_x^1(\Omega,\Omega) & k_c^1(\Omega,\Omega) \\
            k_c^2(\Omega,\Omega) & k_y^1(\Omega,\Omega)
        \end{bmatrix},
        m(\Omega)=
        \begin{bmatrix}
            m_x(\ejo) \\ m_y(\ejo)
        \end{bmatrix},
    \end{align}
    \begin{align}
        k_x(\Omega_1,\Omega_2)
        &= 
        \tfrac{1}{2}\Re{k(\ej{\Omega_1},\ej{\Omega_2}) + \tilde{k}(\ej{\Omega_1},\ej{\Omega_2})}
        \\
        k_y(\Omega_1,\Omega_2)
        &= 
        \tfrac{1}{2}\Re{k(\ej{\Omega_1},\ej{\Omega_2}) - \tilde{k}(\ej{\Omega_1},\ej{\Omega_2})}
        \\
        k_c(\Omega_1,\Omega_2)
        &= 
        \tfrac{1}{2}\Im{\tilde{k}(\ej{\Omega_1},\ej{\Omega_2}) - k(\ej{\Omega_1},\ej{\Omega_2})}
    \end{align}
and where the superscripts denote
derivatives, e.g. 
$
    k_x^{12}=\partial^2 k_x / \partial \Omega_1 \Omega_2.
$
\end{theorem}
\begin{remark}
    If $f$ is conjugate-symmetric, then $|f|$ is an even function of $\Omega$,
    meaning that the probability of a gain bound violation over
    $\Omega\in[-\pi,\pi)$ is equal to the violation probability over the reduced
    range $\Omega\in[0,\pi]$. In this case, we need only compute the expected
    number of gain upcrossings between $0$ and $\pi$, and the lower limit of the
    outer integral in~\eqref{eq:gain_upcrossing_formula} can be changed from
    $-\pi$ to zero.
\end{remark}

\begin{figure}[htpb]
    \centering
    \includegraphics[width=0.49\linewidth]{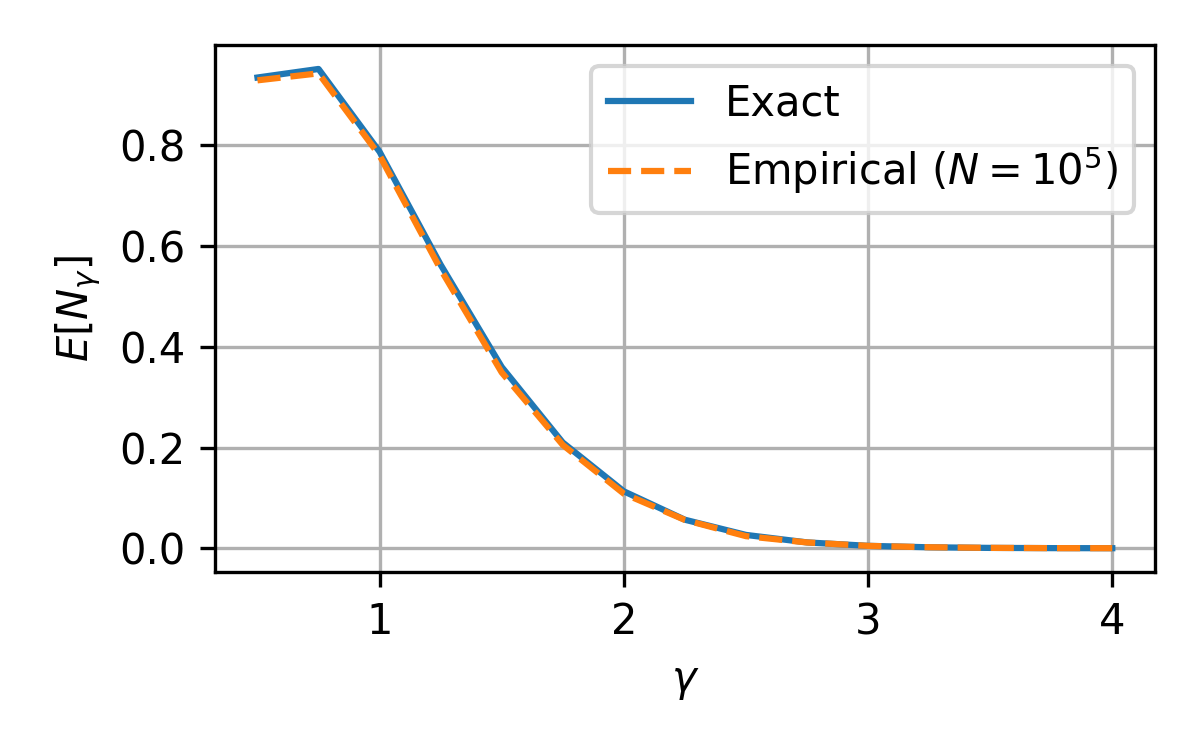}
    \includegraphics[width=0.49\linewidth]{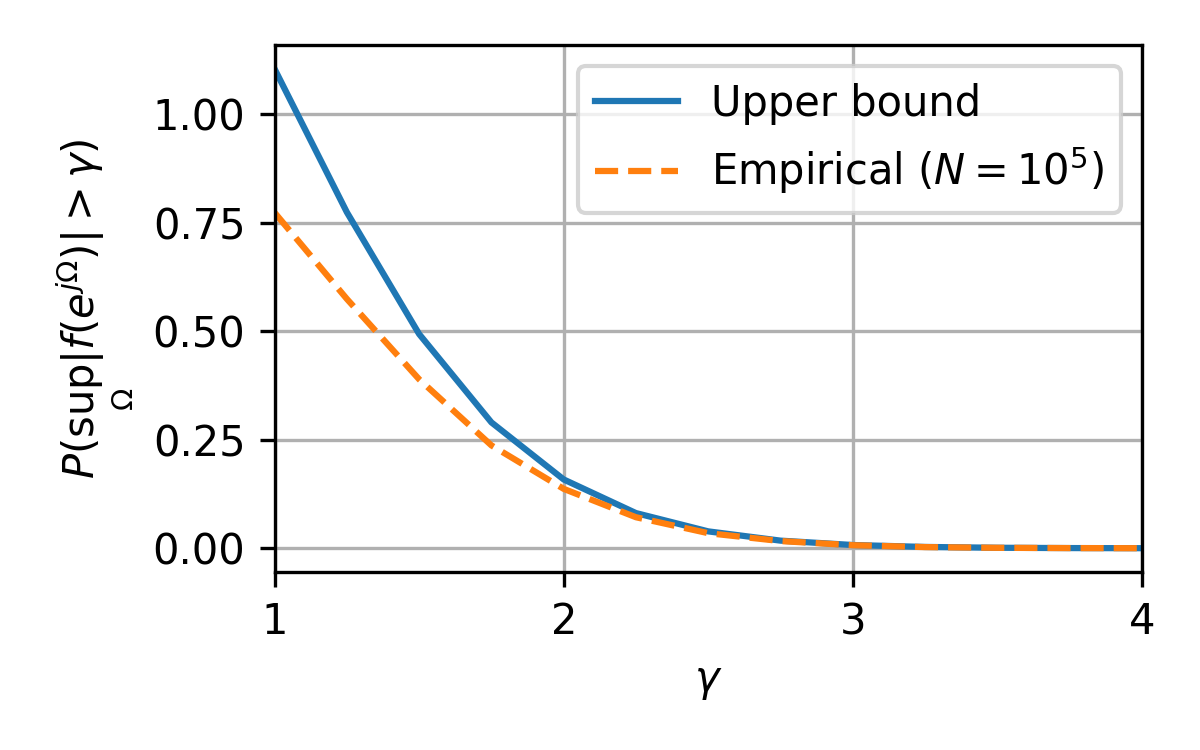}
    \caption{Evaluating the Belyaev formula and the excursion gain bound on a
        geometric $H_\infty$ process with $\alpha=0.5$.
        Left: Numerical comparison of~\eqref{eq:belyaev_formula} with an
        empirical estimate ($N=100,000$) of $\Ex{N_\gamma}$. Right: Numerical comparison
of~\eqref{eq:excursion-bound} with an empirical estimate ($N=100,000$) of $P_\gamma(f)$.}%
    \label{fig:belyaev-experiment} 
\end{figure} 

We next use a simple example to numerically demonstrate the validity of the gain upcrossing formula~\eqref{eq:gain_upcrossing_formula}
and to test the tightness of the bound~\eqref{eq:excursion-bound}. Our example
uses a geometric $H_\infty$ GP $f$ with $\alpha=0.5$
(Example~\ref{ex:geometric}) and a range of threshold values $\gamma\in[1,4]$.
The
expected $H_2$ norm of a geometric $H_\infty$ process is $1$, so it's reasonable
to expect $\Ex{N_\gamma}$ and $P_\gamma(f)$ to be relatively high near $\gamma=1$, and to taper
off relatively quickly. Figure~\ref{fig:belyaev-experiment} shows the results of
the numerical evaluation of
Equations~\eqref{eq:gain_upcrossing_formula},~\eqref{eq:excursion-bound} and compares
the results to empirical approximations of $\Ex{N_\gamma}$ and $P_\gamma(f)$ made using
$N=100,000$ process realizations. We can conclude from the figure
that~\eqref{eq:belyaev_formula} accurately computes $\Ex{N_\gamma}$ as expected.
Furthermore, we see that~\eqref{eq:excursion-bound} indeed upper bounds
$P_\gamma\gamma(f)$; while the bound is conservative in regions where $P_\gamma(f)$ is high, it
quickly becomes tighter where $P_\gamma(f)$ is small, levelling off to overestimate
$P_\gamma(f)$ by $\sim 10\%$ for $\gamma>2.5$. Since we will generally engineer $P_\gamma(f)$
to be small in applications, the experiment shows
that~\eqref{eq:excursion-bound} is not a conservative bound in regions of
practical interest.

\section{Conclusion}%
\label{sec:conclusion}

This paper has three principal contributions: how to construct $H_\infty$ GPs,
how to refine an $H_\infty$ GP model with data, and how to prove robustness
guarantees for systems with probabilistic $H_\infty$ feedback uncertainties. The
fact that $H_\infty$ GPs are amenable both to regression and rigorous robustness
analysis is a critical advantage in applications calling for robustly safe
learning of uncertainties of \emph{a priori} unknown order; such uncertainties
frequently arise when flexible mechanics (e.g. flexible modes, aeroelastic
coupling, soft robotic parts) are neglected in the nominal plant.

The regression-based refinement method described in
Section~\ref{sec:regression} depends strongly on the method used to convert
time-domain data to frequency-domain, many of which (such as filter banks and
ETFEs) are sensitive to measurement noise. A more direct application of
time-domain data to the frequency-domain model, perhaps effected by a
\emph{projected-process regression}~\cite[\S 8.3.4]{gpml}.
Even within the regime of refinement via frequency-domain data, there remains
some mystery around the performance gap between strictly linear and widely
linear estimators.
We conjecture that the similarity in performance is due to the perfect correlation between real and
imaginary parts of the transfer function induced by  via the discrete Hilbert
transform~(see~\cite{rabiner1975theory}, \S 2.26) induced by the symmetry
condition required for the process to give real outputs to real inputs. In the
large-sample limit, this condition causes the prior regression covariance to
become \emph{maximally improper}, a special case in which the strictly linear
estimator is indeed optimal.

A key direction for future work is to extend $H_\infty$ GP regression and
Belyaev-based bounds on excursion probabilities to the MIMO case. MIMO regression
should pose little difficulty; a simple tactic, if no better can be found, is to
perform regression independently on each element of a transfer function matrix.
For MIMO excursion bounds, such an independent approach would likely not be as
effective; instead we would need to directly examine MIMO gains (e.g. by a
probabilistic extension of sigma plots) or direct examination of probabilistic
IQCs. In either case, the near-Gaussian structure enjoyed in the SISO case will
likely not hold, so a more general theory of excursion probabilities will be
necessary.

\section{Proofs of Main Results}%
\label{sec:proofs_of_main_results}

\begin{proof}[Proof of Lemma~\ref{lem:analytic-process-compact}]
    We begin by constructing a proper z-domain process $f$ with covariance $k$ from
    an orthonormal basis for the reproducing kernel Hilbert space $\Hk{k}$.
    Then, we establish that $f$ has a mean-square derivative $f_1$ that has a
    continuous version. Finally, we use $f_1$ and a point evaluation $f(z_0)$ to
    construct a process $g$ that has complex differentiable (i.e. holomorphic
    and therefore analytic) sample paths, which we then show is pointwise
    identical to $f$.

    Since $k$ is \emph{a forteriori} continuous, $\Hk{k}$ is separable and
    admits a countable orthonormal basis $\{\phi_n\}_{n=0}^\infty$, and we can express $k$ as the
    pointwise and uniformly converging sum
    $
        k(z,w) = \sum_{n=0}^\infty \phi_n(z)\phi^*_n(w).
    $
    Let $\{a_n\}_{n=0}^\infty$, $\{b_n\}_{n=0}^\infty$ denote two sequences of
    independent standard normals, that is
    $a_1,b_1,a_2,b_2,\dotsc\simiid\Norm{0,1}$. Then the random function
    $
        f(z) = \sum_{n=0}^\infty \frac{1}{\sqrt{2}}(a_n + jb_n)\phi_n(z)
    $
    is a z-domain Gaussian process on the set of $z$ where the sum converges.
    From the summation form of $f$ we can verify
            $\Ex{f(z)f(w)^*} = k(z,w), $
            $\Ex{f(z)f(w)} = 0,$
    making $f$ a proper z-domain process where it exists; and since $k(z,z)$ is
    finite for all $z\in D$, it follows that the pointwise sums $f(z)$ converge
    for all $z\in D$.

    We now prove that $f$ admits a version $g$ whose realizations are analytic.
    To do so, we demonstrate that $f$ admits a mean-square derivative $f_1$
    which admits a continuous version, which we then use to construct a version
    of $f$ that is holomorphic and therefore analytic.

    For the real-domain case, it is well known that a stochastic process
    possesses a mean-square derivative if its covariance function is twice
    differentiable, and that the covariance of the mean-square derivative is
    the first derivative of the covariance with respect to both
    arguments~\cite[\S 34]{loeve2017probability}.
    The complex-domain case is nearly the same, except that the derivative with respect
    to the second argument must be a \emph{conjugate derivative}. Specifically,
    suppose that the limit
    \begin{equation}
        \begin{aligned}
            &\lim_{z\to z_0, w\to w_0}
            \Ex{
                \left(\frac{f(z) - f(z_0)}{z-z_0}\right)
                \left(\frac{f(w) - f(w_0)}{w-w_0}\right)^*
            }\\
            &=\lim_{z\to z_0, w\to w_0}
            \frac{k(z,w) - k(z,w_0) - k(z_0,w) + k(z_0,w_0)}{(z-z_0)(w-w_0)^*}
            = \frac{\partial^2}{\partial z \partial w^*} k
        \end{aligned}
    \end{equation}
    exists for all $z_0, w_0 \in D$: then any proper z-domain process $f_1$ with
    Hermitian covariance $\frac{\partial^2}{\partial z \partial w^*} k$ is a
    mean-square derivative of $f$,\footnote{The fact that
    $\frac{\partial^2}{\partial z \partial w^*}$ k is indeed a valid
    Hermitian covariance is established by the fact that positive definite functions
    are closed under limits.}
    since it satisfies the mean-square differentiability criterion
        $\lim_{z\to z_0}
        \Ex{
            \left| \frac{f(z) - f(z_0)}{z-z_0} \right| - f_1(z_0)
        } = 0$.
    By our assumption that $k$ is holomorphic in the first argument and
    antiholomorphic in the second argument, it follows that $k$ possesses derivatives in the
    first argument, and conjugate derivatives in the second argument, of all
    orders; let $k_{nm}$ denote the derivative
    of $k$ taken $n$ times in the first argument and the conjugate derivative $m$ times in the second
    argument. In particular, the existence of $k_{11}$ means that
    $f$ indeed possesses a mean square derivative $f_1$ as described above.

    Since $f_1$ is mean zero and proper, it follows that its real and imaginary
    parts ($x_1$ and $y_1$ respectively) are independent and identical real Gaussian processes with mean zero,
    and calculation from the relation $x_1 = \frac{1}{2}(f_1+f_1^*)$ yields the
    covariance
    \begin{equation*}
        \Ex{x_1(z)x_1(w)} = \Re{k_{11}(z,w)}.
    \end{equation*}
    From this it follows that
    $x_1$ has the canonical metric
    \begin{equation*}
        d^2_{x_1}(z,w) = \Ex{(x_1(z) - x_1(w))^2}
        = k_{11}(z,z) + k_{11}(w,w) - 2k_{11}(z,w).
    \end{equation*}
    By assumption we then have
    $
        d^2_{x_1}(z,w) \le |z - w|^{2+\alpha},
    $
    and from~\cite[prop. 1.16]{azais2009level} it follows that $x_1$ admits a version with
    continuous sample paths. Evidently the same is true for $y_1$, being
    identically distributed to $x_1$, which means that $f_1$ admits a continuous
    version.
    Define $g$ as the z-domain process whose realizations are formed as
    $
        g_\xi(z) = f(z_0) + \int_\gamma f_{1,\xi}(\zeta)d\zeta,
    $
    where $f_{1,\xi}$ is a realization of the continuous version of $f_1$,
    $z_0\in D$ and $\gamma$ is a continuous path starting at $z_0$ and
    ending at $z$. The realizations $g_\xi$ evidently possess derivatives; by
    the fundamental theorem of calculus we have
        $\frac{d}{dz} g_\xi = f_{1,\xi}$,
    which is continuous and therefore bounded on the compact domain $D$. We now
    show that $g$ is in fact a version of $f$ by showing that
    \begin{equation*}
        \begin{aligned}
            \Ex{|f(z) - g(z)|^2}
            &= \Ex{f(z)f(z)^* - f(z)g(z)^* - g(z)f(z)^* +
            g(z)g(z)^*}\\
            &= k(z,z) - \Ex{f(z)g(z)^*} - \Ex{g(z)f(z)^*} +
            \Ex{g(z)g(z)^*}
            = 0.
        \end{aligned}
    \end{equation*}
    For the second term we have
    \begin{equation*}
        \begin{aligned}
            \Ex{f(z)g(z)^*}
            &=
            \Ex{
                f(z)f(z_0)^*
                + f(z)\int_\gamma f_1(\zeta)^*d\zeta^*
            }\\
            &= \Ex{f(z)f(z_0)^*}
            + \int_\gamma \Ex{f(z)f_1(\zeta)}d\zeta^*
            = k(z,z_0) + \int_\gamma k_{01}(z,\zeta)d\zeta^*\\
            &= k(z,z),
        \end{aligned}
    \end{equation*}
    where the last line follows from the fact that $k$ is a complex
    antiderivative for $k_{01}$. The third term immediately follows since
    \begin{equation*}
    \Ex{g(z)f(z)^*} = \Ex{f(z)g(z)^*}^*
    = k(z,z)^* = k(z,z).
    \end{equation*}
    For the final term we have
    {\allowdisplaybreaks[4]
    \begin{equation*}
        \begin{aligned}
            \Ex{g(z)g(z)^*}
            &=
            \Ex{
                f(z_0)f(z_0)^*
            + f(z_0) \int_\gamma f_1(\zeta)^*d\zeta^*}\\
            &\quad+ \Ex{f(z_0)^* \int_\gamma f_1(\zeta)d\zeta
                + \left(\int_\gamma f_1(\zeta)d\zeta \right)
                \left(\int_\gamma f_1(\zeta)^*d\zeta^* \right)
            }\\
            &=
            \Ex{f(z_0)f(z_0)^*}
            + \int_\gamma \Ex{f(z_0)f_1(\zeta)^*} d\zeta^*
            + \int_\gamma \Ex{f(z_0)^*f_1(\zeta)} d\zeta\\
            &\quad+ \int_\gamma\int_\gamma \Ex{f(\zeta_1)f(\zeta_2)^*}d\zeta_1 d\zeta_2^*\\
            &=
            k(z_0,z_0)
            + \int_\gamma k_{01}(z_0,\zeta) d\zeta
            + \int_\gamma k_{01}(z_0,\zeta)^* d\zeta
            + \int_\gamma\int_\gamma k_{11}(\zeta_1,\zeta_2)d\zeta_1 d\zeta_2\\
            &=
            k(z_0,z_0)
            + (k(z_0,z) - k(z_0,z_0))
            + (k(z_0,z)^* - k(z_0,z_0)^*)\\
            &+ (k(z,z) - k(z,z_0) - k(z_0,z) + k(z_0,z_0))
            = k(z,z).
        \end{aligned}
\end{equation*}}
    Note that the exchange of the double integral is permitted by Fubini's theorem
    since realizations of $f_1$ are bounded. We have also used the Hermitian
    covariance properties $k(z,w)=k(w,z)^*$,
    $k(z,z)=k(z,z)^*$. Putting the terms together, we arrive at
    \begin{equation*}
        \begin{aligned}
            \Ex{|f(z) - g(z)|^2}
            &= k(z,z) - k(z,z) - k(z,z) + k(z,z) = 0,
        \end{aligned}
    \end{equation*}
    proving that $g$ is a version of $f$.
\end{proof}

\begin{proof}[Proof of Theorem~\ref{prop:stationary-process}]
    First, we show that $f$ having the form~\eqref{eq:process_expansion} with
    positive
    $\{a_n\}_{n=0}^\infty\in\ell^1$ implies
    that $f$ is a conjugate-symmetric and Hermitian stationary $H_\infty$ process
    with the given covariances.
    Suppose that $f(z) = \sum_{n=0}^\infty a_n w_n z^{-n}$.
    We can readily see that
    \begin{equation}
        \begin{aligned}
            k(z,w)
            &= \Ex{f(z)f(w)^*}
            = \Ex{
                \left(
                    \sum_{n=0}^\infty a_n w_n z^{-n}
                \right)
                \left(
                    \sum_{m=0}^\infty a_m w_m (w)^{-m}
                \right)^*
            }\\
            &= \Ex{
                \sum_{n=0,m=0}^\infty
                a_n a_m w_n w_m z^{-n} (w^*)^{-n}
            }
            = \sum_{n=0}^\infty a_n^2 (zw^*)^{-n},
        \end{aligned}
    \end{equation}
    where the cross terms vanish by the independence of the $w_n$.
    Conjugate symmetry follows from the direct calculation
    $
    f(z)^*
    =
    (\sum_{n=0}^\infty a_n w_n z^{-n})^*
    =
    \sum_{n=0}^\infty a_n w_n(z^*)^{-n}
    =
    f(z^*)
    $.
    Since $f$ is conjugate symmetric, its realizations represent systems with a real-valued
    impulse response when interpreted as transfer functions.
    Recall that a SISO LTI system is BIBO stable if and only if its impulse response
    is absolutely summable.
    To that end, consider the sequence
    $
        M_T=
        \sum_{n=0}^T |h_f(n)|
        =\sum_{n=0}^T a_n |w_n|
    $
    of partial sums:
    if $\lim_{T\to\infty}M_T$ converges to a random variable that is finite with probability one, then the impulse
    response is absolutely summable with probability one. Since the $w_n$ are independent and
    $0 < a_n |w_n| < \infty $ for all $n$, it follows that $M_T$ is a submartingale and that
    $\Ex{M_T}$ increases monotonically. Using the summability condition on the
    $a_n$ and the fact that $\Ex{|w_n|}=\sqrt{2/\pi}$ (as $|w_n|$ follows a
    half-normal distribution), we have
    \begin{equation}
        \Ex{\sum_{n=0}^\infty |h_f(n)|} = \sum_{n=0}^\infty a_n |w_n|
        =\sqrt{\frac{2}{\pi}}\sum_{n=0}^\infty a_n<\infty,
    \end{equation}
    which means $\sup_T \Ex{M_T}<\infty$ by monotonicity. Since $M_T$ is a
    submartingale and $\sup_T \Ex{M_T}$ is finite, it follows by the
    Martingale convergence theorem~\cite[Theorem 4.2.11]{durrett2019probability}
    that the limit of $M_T$ converges to a
    random variable that is finite with probability one.
    This shows that $h_f$ is absolutely summable with probability one, implying
    BIBO stability and that $f\in H_\infty$ with probability one.

    Next, we show that a Hermitian stationary, conjugate-symmetric $H_\infty$
    Gaussian process $f$ must have Hermitian and complementary covariances
    of the form~\eqref{eq:covariance_expansion}, and that this in turn implies
    that $f$ has the form~\eqref{eq:process_expansion}. Since $f\in H_\infty$,
    we can use the fact that $z^{-n}, n \ge 0$ is a basis for $H_\infty$ to
    expand $f$ as
    $
        f(z) = \sum_{n=0}^\infty h_n z^{-n},
    $
    where the coefficients $h_n = \langle f,z^{-n}\rangle_2$ are an infinite
    sequence of random variables.
    Since $f$ is Gaussian and conjugate symmetric,
    the $h_n$ are real Gaussian random
    variables that may be correlated. From this form, we can express the
    Hermitian and complementary covariance as
    \begin{equation}
            k(z,w) = \sum_{n=0}^\infty \sum_{m=0}^\infty
                \Ex{h_n h_m} z^{-n}(w^*)^{-m},
    \end{equation}
    \begin{equation}
            \tilde{k}(z,w) = \sum_{n=0}^\infty \sum_{m=0}^\infty
                \Ex{h_n h_m} z^{-n}(w)^{-m},
    \end{equation}
    which shows that $\tilde{k}(z,w)=k(z,w^*)$ for $z,w\in E^2$. Restricting the
    covariance functions to the unit circle, we have
    \begin{equation}
        \begin{aligned}
            \label{eq:covars1}
            k(e^{j\theta},e^{j\phi}) &= \sum_{n=0}^\infty \sum_{m=0}^\infty
                \Ex{h_n h_m} e^{-j(n\theta-m\phi)}\\
            \tilde{k}(e^{j\theta},e^{j\phi}) &= \sum_{n=0}^\infty \sum_{m=0}^\infty
                \Ex{h_n h_m} e^{-j(n\theta+m\phi)}.\\
        \end{aligned}
    \end{equation}
    By the assumption of Hermitian stationarity, we know that
    $k(e^{j(\theta-\phi)}, 1)$ is a positive definite function whose domain is
    the unit circle. We can therefore apply Bochner's
    theorem~\cite[section 1.4.3]{rudin1962fourier} to obtain a second expansion
    \begin{equation}
            \label{eq:covars2}
        k(e^{j(\theta-\phi)}, 1)=
        k(e^{j\theta}, e^{j\phi})=
        \sum_{n\in\Z} a_n^2 e^{-jn(\theta - \phi)},
    \end{equation}
    where $a_n$ are real and nonnegative. In order for the
    expansion of $k$ in~\eqref{eq:covars1} and the expansion
    in~\eqref{eq:covars2} to be equal, the positive-power terms
    in~\eqref{eq:covars2} must vanish, and the cross-terms in~\eqref{eq:covars1}
    must vanish.

    This means that $\Ex{h_n h_m}=0$ for $m\ne n$,
    from which it follows that the covariances have the form
    \begin{equation}
        \begin{aligned}
            k(z,w) &=
            \sum_{n=0}^\infty
            \Ex{h_n^2} (zw^*)^{-n}
            =
            \sum_{n=0}^\infty
            a_n^2 (zw^*)^{-n}
            \\
            \tilde{k}(z,w) &=
            \sum_{n=0}^\infty
            \Ex{h_n^2} (zw)^{-n}
            =
            \sum_{n=0}^\infty
            a_n^2 (zw)^{-n}
        \end{aligned}
    \end{equation}
    where we identify $a_n^2=\Ex{h_n^2}$,
    and that the $h_n$ are independent. Returning to the expanded form of the
    process and expressing $\Ex{h_n^2}=a_n^2$, we have
    $
        f(z) = \sum_{n=0}^\infty w_n a_n z^{-n}
    $
    where $w_n\simiid\Norm{0,1}$.

    Evidently, the impulse response $h_f$ has the same form as
    before, so the expected absolute sum of the impulse response is
    $\Ex{\sum_{n=0}^\infty |h_f(n)|}=\sqrt{\frac{2}{\pi}}\sum_{n=0}^\infty a_n$.
    Since $f\in H_\infty$ by assumption, it follows that
    $\sum_{n=0}^\infty |h_f(n)|$ almost surely converges, and therefore that
    $\Ex{\sum_{n=0}^\infty |h_f(n)|} <\infty$ by the Kolmogorov three-series
    theorem~(\cite[Theorem 2.5.8]{durrett2019probability}, condition (ii)),
    showing that $\{a_n\}_{n=0}^\infty\in\ell^1$.

\end{proof}

\begin{proof}[Proof (of Theorem~\ref{thm:gain_upcrossing_formula})]

To apply Belyaev's formula,
we must first establish that the conditions of Theorem~\ref{thm:belyaev_formula}
are satisfied by the vector GP $g$ composed of the real and imaginary parts of
$f$. That $g(\ejo)$ and $g'(\ejo)|g(\ejo)$ have distributions is
given by the fact that $g$ is a Gaussian process: $g(\ejo)$ is
Gaussian-distributed by definition; the derivative of a Gaussian process is
itself a Gaussian process; and a Gaussian random variable conditioned on another
Gaussian random variable is itself Gaussian. The given conditions on
differentiability ensure that $g$ is differentiable: a mean-zero process with
thrice-differentiable covariances is at least once-differentiable, and adding a
differentiable mean to the process preserves this property.
Finally, the
surface $x^2+y^2-\gamma^2=0$ satisfies the third condition, as an
$\epsilon$-neighborhood of $S_\phi$ is simply an annulus in the plane with inner
radius $\gamma-\epsilon$ and outer radius $\gamma+\epsilon$, which can be parameterized
using polar coordinates. 

With these conditions established, we know that the expected number of $\gamma$-level
gain upcrossings of $f$ is given by~\eqref{eq:belyaev_formula} with $g=(x,y)$
and $\Phi(x,y)=x^2+y^2-\gamma^2$; all that remains is to show how to
express~\eqref{eq:belyaev_formula} in terms computable from $k$, $\tilde{k}$,
and $m$, that is to derive~\eqref{eq:gain_upcrossing_formula}.
First, there is the matter of integration over $S_\Phi$: this can be handled by
integrating over its circular
parameterization $S_\Phi=\{(\gamma \cos\theta, \gamma \sin\theta), \theta\in[-\pi,\pi)\}$ 
in which case 
$
    z=z(\theta)=(\gamma\sin\theta,\gamma\cos\theta),\quad ds(z)=\gamma d\theta.
$
To obtain the distribution $p_{g(\Omega)}$, we require the first- and
second-order statistics of $g$. Since the real and imaginary parts of $f$ are
Gaussian, it follows that $g$ is a vector Gaussian process. Its mean and
variance are
\begin{equation}
    \Ex{g(\ejo)} 
    =
    \begin{bmatrix}
        \Ex{x(\ejo)} \\ \Ex{y(\ejo)}
    \end{bmatrix} = m(\Omega),
    \quad
        \Sigma(\Omega_1,\Omega_2)= 
        \begin{bmatrix}
            k_x(\Omega_1,\Omega_2) & k_c(\Omega_1,\Omega_2)\\
            k_c(\Omega_1,\Omega_2) & k_y(\Omega_1,\Omega_2)\\
        \end{bmatrix},
\end{equation}

which can be verified by substituting
$
        x(\ejo)=\frac{1}{2}(f(\ejo)+f(\ejo)^*),
$
$
        y(\ejo)=\frac{1}{2j}(f(\ejo)-f(\ejo)^*),
$
and working out the expectations.
From this, it follows that
\begin{equation}
\label{eq:upcrossing_proof_distribution}
p_{f(\Omega)}(z) = (2\pi)^{-1}\Sigma(\Omega,\Omega)e^{-\tfrac{1}{2} (z -
m(\Omega))^\top \Sigma(\Omega,\Omega)(z-m(\Omega))}.
\end{equation}
To compute the conditional expectation, we also require the joint distribution
of $g(\Omega)$ and its derivative $g^\prime(\Omega)$: this is another
multivariate normal; this is
\begin{equation}
    \begin{bmatrix}
        g^\prime(\Omega) \\ g(\Omega)
    \end{bmatrix}
    \sim
    \Norm{
        \begin{bmatrix}
            m_x^\prime(\Omega) \\
            m_y^\prime(\Omega) \\
            m_x(\Omega) \\
            m_y(\Omega)
        \end{bmatrix},
        \begin{bmatrix}
            \Sigma^\prime(\Omega,\Omega) & C(\Omega,\Omega) \\
            C^\top(\Omega,\Omega) & \Sigma(\Omega,\Omega)
        \end{bmatrix}
    }.
\end{equation}
We next compute the conditional distribution of $g$ given $g^\prime=z$; this is
again a multivariate normal: dropping $\Omega$ from the submatrices for brevity,
we have
\begin{equation}
    g(\Omega) | g^\prime(\Omega)=z
    \sim
    \Norm{
        m^\prime + C\Sigma^{-1}(z-m), 
        \Sigma^\prime - C\Sigma^{-1}C^\top
    }.
\end{equation}
The normal vector $n_\Phi(z)$ for $\Phi$---a circle with center zero and radius
$\gamma$---is simply $\gamma^{-1} z$. 
Thus the linear mapping 
$
(g(\Omega) | g^\prime(\Omega)=z)
\mapsto
n_\Phi(z)^\top(g(\Omega) | g^\prime(\Omega)=z)
$
gives us the distribution
\begin{equation}
    \begin{aligned}
        &n_\Phi(z)^\top(g(\Omega) | g^\prime(\Omega)=z)\\
        &\quad\sim
        \Norm{\gamma^{-1}z^\top m^\prime + \gamma^{-1}z^\top C\Sigma^{-1}(z-m), 
            \gamma^{-2} z^\top(\Sigma^\prime - C(\Sigma^{-1}C^\top) z}\\
        &\quad \eqdef \Norm{\mu_z(\Omega), \sigma^2_z(\Omega)}.
    \end{aligned}
\end{equation}

The remaining step is to compute the expectation of the positive part. Since
$
    n_\Phi(z)^\top(g(\Omega) | g^\prime(\Omega)=z)
$
is a Gaussian scalar, the
positive part is a rectified Gaussian, whose mean is
\begin{equation}
    \begin{aligned}
        \label{eq:upcrossing_proof_expectation}
        \Ex{(n_\Phi(z)^\top(g(\Omega))_+ | g^\prime(\Omega)=z)}
        &=
        \frac{\sigma_z(\Omega)}{\sqrt{2\pi}}e^{-\frac{1}{2}(\mu_z(\Omega)/\sigma_z(\Omega))^2}\\
        &+ \frac{1}{2}\mu_z(\Omega)\left(
                1 + \erf\left(\frac{\mu_z(\Omega)}{\sqrt{2}\sigma_z(\Omega)}\right)
            \right).
    \end{aligned}
\end{equation}
Finally, applying~\eqref{eq:upcrossing_proof_distribution}
and~\eqref{eq:upcrossing_proof_expectation} to the Belyaev
formula~\eqref{eq:belyaev_formula} yields the expression~\eqref{eq:gain_upcrossing_formula}. 
\end{proof}

\bibliographystyle{siamplain}
\bibliography{refs}

\end{document}